\documentclass[preprint,12pt,3p]{elsarticle}

\usepackage{amssymb}

\usepackage{amsthm}
\usepackage{graphicx}
\usepackage{amsmath}
\usepackage{amsfonts}
\usepackage{hyperref}

\usepackage{multirow}
\usepackage{latexsym}
\usepackage{setspace}
\PassOptionsToPackage{boxed,section}{algorithm}
\usepackage{algorithm}
\usepackage{algorithmic}
\usepackage{enumerate}
\usepackage{amsmath}			
\usepackage{amssymb}			
\usepackage{url}
\usepackage{setspace}
\usepackage{geometry}
\usepackage{bbm}
\newcounter{mynum} \numberwithin{mynum}{section}

\newtheorem{theorem}[mynum]{\em Theorem}
\newtheorem{proposition}[mynum]{\em Proposition}
\newtheorem{conjecture}[mynum]{\em Conjecture}
\newtheorem{definition}[mynum]{\em Definition}
\newtheorem{lemma}[mynum]{\em Lemma}

\newtheorem{corollary}[mynum]{\em Corollary}
\newtheorem{claim}[mynum]{\em Claim}

\newtheorem{question}[mynum]{\em Question}
\newtheorem{observation}[mynum]{\em Observation}

\journal{Sample Journal}

\begin{document}

\begin{frontmatter}

\title{Vizing-Goldberg type bounds for the equitable chromatic number \\of  block graphs\footnote{The work of the second author has been partially supported by Narodowe Centrum Nauki under contract 2017/01/X/ST6/00148}$^{,}$\footnote{The work of the third author has been partially supported by the Italian MIUR PRIN 2017 Project ALGADIMAR ``Algorithms, Games, and Digital Markets.''}}

\author[label1]{Janusz Dybizba\'nski}
\address[label1]{Institute of Informatics, Faculty of Mathematics, Physics and Informatics,\\University of Gda\'nsk, Wita  Stwosza 57, Gda\'nsk, Poland}

\cortext[cor1]{Corresponding author}

\ead{jdybiz@inf.ug.edu.pl}

\author[label1]{Hanna Furma\'nczyk\corref{cor1}}
\ead{hanna.furmanczyk@ug.edu.pl}

\author[label4,label5]{Vahan Mkrtchyan}
\address[label4]{Dipartimento di Informatica, Universita degli Studi di Verona, Verona, Italy}
\address[label5]{Gran Sasso Science Institute, L'Aquila, Italy}
\ead{vahan.mkrtchyan@gssi.it}

\begin{abstract}
An equitable coloring of a graph $G$ is a proper vertex coloring of $G$ such that the sizes of any two color classes differ by at most one. In the paper, we pose a conjecture that offers a gap-one bound for the smallest number of colors needed to equitably color every block graph. In other words, the difference between the upper and the lower bounds of our conjecture is at most one. Thus, in some sense, the situation is similar to that of chromatic index, where we have the classical theorem of Vizing and the Goldberg conjecture for multigraphs. The results obtained in the paper support our conjecture. More precisely, we verify it in the class of well-covered block graphs, which are block graphs in which each vertex belongs to a maximum independent set. We also show that the conjecture is true for block graphs, which contain a vertex that does not lie in an independent set of size larger than two. Finally, we verify the conjecture for some symmetric-like block graphs. In order to derive our results we obtain structural characterizations of block graphs from these classes.  

\end{abstract}


\end{frontmatter}



\section{Introduction}
A \emph{hypergraph} is a pair $\mathcal{H}=(V,\mathbb{E})$, where $V$ is an $n$ element set of vertices of $H$ and $\mathbb{E}$ 
is a family of $m$ non-empty subsets of $V$ called edges or \emph{hyperedges}.
Let $[k]$ denote the set of integers $\{1,\ldots,k\}$.
A \emph{$k$-coloring} of hyperedges of $\mathcal{H}=(V,\mathbb{E})$ is a mapping $c : \mathbb{E} \rightarrow [k]$ such that no two edges that
share a vertex get the same color (number). An edge $k$-coloring of $\mathcal{H}=(V,\mathbb{E})$ is \emph{equitable} if each color class is of size $\lceil m/k \rceil$ or $\lfloor m/k \rfloor$. In other words, an equitable edge coloring of $\mathcal{H}=(V,\mathbb{E})$ may be seen as a partition of the hyperedge set $\mathbb{E}$ into independent subsets $E_1, \ldots, E_k$ such that $||E_i|-|E_j||\leq 1$, for each $i,j \in [k]$. The smallest 
$k$ such that $\mathcal{H}$ admits an equitable edge $k$-coloring is called the \emph{equitable chromatic index} and is denoted by $\chi'_=(\mathcal{H})$.

For a hypergraph $\mathcal{H}$ we need to define the concept of its line graph/host graph (cf. Fig.~\ref{fig:host}).
The \emph{line graph} $L(\mathcal{H})$ is a simple graph representing adjacencies between hyperedges in $\mathcal{H}$. More
precisely, each hyperedge of $\mathcal{H}$ is assigned a vertex in $L(\mathcal{H})$ and two vertices in $L(\mathcal{H})$ are adjacent if and only if their
corresponding hyperedges share a vertex in $\mathcal{H}$.
We say that a hypergraph $\mathcal{H}$ has an \emph{underlying (host)} graph $G$ (spanned on the same set of vertices) if each hyperedge of $\mathcal{H}$ induces a connected subgraph in $G$. Furthermore, it is assumed that for each 
edge $e_G$ in $G$ there exists a hyperedge 
$e_{\mathcal{H}}$ in $\mathcal{H}$ such that $e_G \subseteq e_{\mathcal{H}}$.

It is easy to notice that an edge coloring of a hypergraph is equivalent to
a vertex coloring of its line graph. A $k$-coloring of vertices of a simple graph $G=(V,E)$ is an assigning of colors from the set $[k]$ to vertices in such a way that no two adjacent vertices receive the same color. A vertex $k$-coloring is \emph{equitable} if each color class is of size $\lceil|V|/k\rceil$ or $\lfloor|V|/k\rfloor$.
The smallest $k$ such that $G$ admits an equitable vertex coloring is called the \emph{equitable chromatic number} of $G$ and is denoted by $\chi_=(G)$. Moreover, note that for a general graph $G$ if it admits an equitable vertex $t$-coloring it does not imply that it admits an equitable vertex $(t+1)$-coloring (cf. for example $t=2$ and $G=K_{3,3}$). That is why we also consider the concept of \emph{equitable chromatic spectrum}, i.e. the set of colors  admitting equitable vertex coloring of the graph. The smallest $k$ such that $G$ admits an equitable vertex $t$-coloring for every $t \geq k$ is called the \emph{equitable chromatic threshold} and is denoted by $\chi_=^*(G)$. If $\chi_=^*(G)=\chi_=(G)$ then we say that the equitable chromatic spectrum of $G$ is \emph{gap-free}.

\begin{figure}
    \centering
    (a)\includegraphics[scale=0.7]{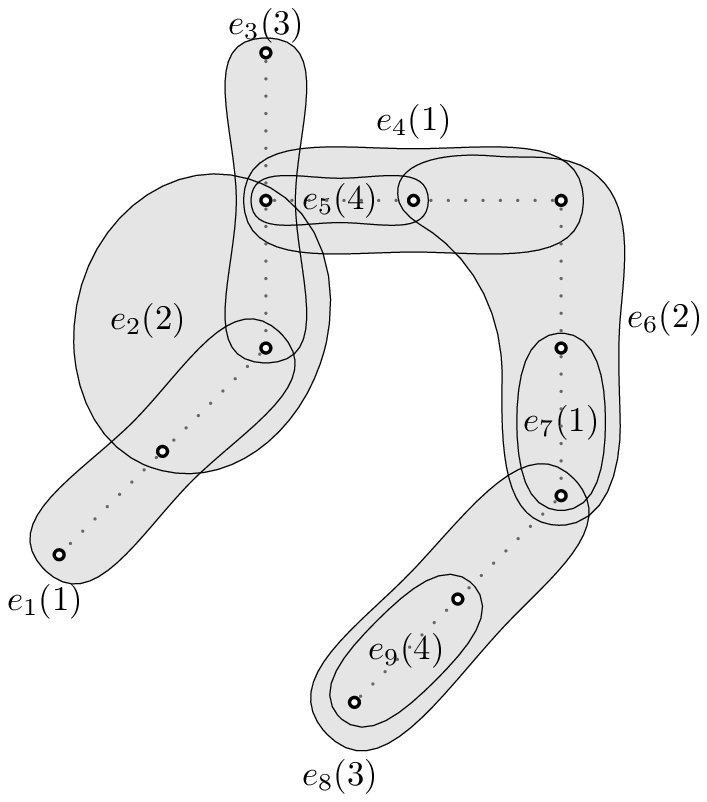}\hspace{0.5cm} (b)\includegraphics[scale=0.7]{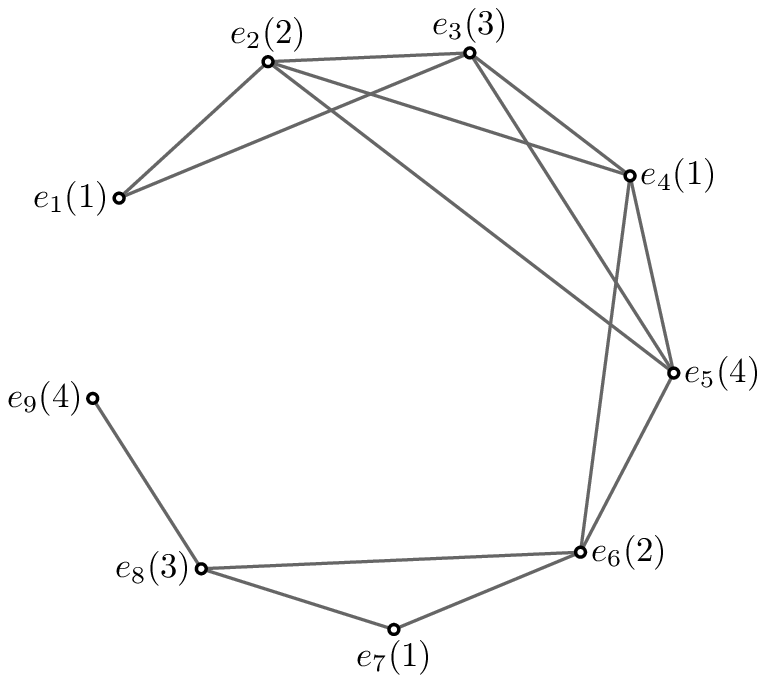} 
    \caption{(a) An example of hypertree $\mathcal{H}$ (hyperedges depicted with solid lines) with an exemplary equitable edge 4-coloring (colors in brackets) and an exemplary host graph (dotted edges). (b) The corresponding line graph $L(\mathcal{H})$ with the corresponding equitable vertex 4-coloring.}
    \label{fig:host}
\end{figure}

Despite the fact that the corresponding problem for simple graphs has been widely studied (for interesting surveys, see \cite{furm:en_book}, \cite{lih}), its generalization to hypergraphs does not seem to have been addressed in the literature. 
To the best of our knowledge there is no paper in the literature that concerns the problem of equitable edge coloring of hypergraphs with the definition given above. 
Hypergraphs in general are very useful in real-life problems modeling, for example in chemistry, telecommunications, and many other fields of science and engineering \cite{hyper}. 
They have also applications in image representation \cite{image}.
Thus, generalization of equitable coloring of simple graphs to hypergraphs seems to be justified. It is known that the model of equitable coloring of simple graphs has many applications, among others in task scheduling (see \cite{furm:4sch}, \cite{obsz:jastrz}). Every time, when we have to divide a system with binary conflict relations into equal or almost equal conflict-free subsystems, we can model this situation by means of equitable graph coloring.

In the paper we study chordal graphs, their subclasses, and the complexity status of the problem of equitable coloring for them. A graph is \emph{chordal} if every cycle of length at least 4 has a chord. It is known (cf. \cite{duchet, mckee}) that a graph $G$ is chordal, if and only if it is a line graph of a hypertree, where \emph{hypertree} is defined as a hypergraph that has an underlying  tree. Thus equitable edge coloring of hypertrees is equivalent to equitable vertex coloring of chordal graphs (cf. Fig.~\ref{fig:host}).  On the other hand, we know (cf.  \cite{bounded}, \cite{bod:part}) that the problem of equitable vertex coloring of interval graphs is NP-hard. Since each interval graph is also chordal, we have also NP-hardness of the problem for chordal graphs. In consequence:

\begin{corollary}
The problem of an equitable edge-coloring for hypertrees is NP-hard.
\end{corollary}

Bodlaender \cite{bounded} proved that the problem of equitable $k$-coloring can be solved in polynomial time for graphs with given tree decomposition and for fixed $k$. The treewidth of a chordal graph equals the maximum clique size minus one. Bodlaender \cite{bounded} proved also that the problem of an equitable $k$-coloring is solvable in polynomial time for graphs with bounded degree even if $k$ is a variable.
\begin{corollary}
The problem of an equitable $k$-coloring is solvable in polynomial time for chordal graphs with bounded maximum clique size.
\end{corollary}
On the other hand, Gomes et al. \cite{EqParamFPT} proved that, when the treewidth is a parameter to the algorithm, the problem of equitable vertex coloring is W[1]-hard. Thus, it is unlikely that there exists a polynomial time algorithm independent of this parameter. In this paper, we address the problem in block graphs, which are the graphs with every 2-connected component being a clique. A clique of a graph $G$ is a maximal complete subgraph of $G$. For block graphs, it is shown in \cite{EqParamFPT} that the problem is W[1]-hard with respect to the treewidth, diameter and the number of colors. This in particular means that under the standard assumption FPT$\neq$W[1] in parameterized complexity theory, the problem is not likely to be polynomial time solvable in block graphs.

In what follows when we refer to equitable coloring we mean equitable vertex coloring unless stated differently.
For a graph $G$ let $\alpha(G)$ be the size of the largest independent set in $G$, while $\alpha(G,v)$ is the size of the largest independent set that contains the vertex $v$ in $G$. Define: 
\[\alpha_{min}(G)=\min_{v\in V(G)}\alpha(G,v).\]
Clearly, $\alpha_{min}(G)\leq \alpha(G)$, and $\alpha_{min}(G)= \alpha(G)$ if and only if every vertex of $G$ lies in a maximum independent set of $G$. Such graphs are known in the literature as \emph{well-covered} graphs \cite{plummer}. A \emph{simplicial vertex} is a vertex that appears in exactly one clique of a graph. A \emph{cut vertex} in a connected graph $G$ is a vertex $v$ that $G-v$ is disconnected. 
In a block graph we define a clique as \emph{pendant} if it contains exactly one cut-vertex, while a clique is \emph{internal} if all its vertices are cut-vertices. Vertices of an internal clique are called \emph{internal} vertices.
For a graph $G$ let $\omega(G)$ be the size of the largest clique of $G$. For every graph, not necessarily a block graph, it is known that
\begin{equation}
    \chi_{=}(G)\geq \max \left\{\omega(G),\left \lceil \frac{|V(G)|+1}{\alpha_{min}(G)+1}\right\rceil\right\}.
    \label{lower_bound}
\end{equation}

Indeed, the equitable chromatic number of a graph $G$ cannot be less than its clique number. Moreover, it cannot be less than $\lceil \frac{|V(G)|+1}{\alpha_{min}(G)+1}\rceil$. The latter follows from the assumption that one color is used exactly $\alpha_{\min}(G)$ times, and any other color can be used at most $\alpha_{\min}(G)+1$ times.
It turns out that the number of colors given by the expression on the right side of the inequality is not sufficient to color equitably every block graph. 
For example, take a clique of size $k$, $k\geq 2$, and add $k+1$ pendant cliques of size $k+1$ to each vertex (cf. Fig. \ref{fig:k=2case}).
It can be easily checked that 
$|V|=k+k(k+1)k=k(k^2+k+1)$,
$\omega(G)=k+1$,
$\alpha(G)=k(k+1)=k^2+k$, while
$\alpha_{min}(G)=1+(k-1)(k+1)=k^2$.
Observe that
$\alpha(G)-\alpha_{min}(G)=k\geq 2.$ 
Thus,
\[\max \left\{\omega(G),\left\lceil \frac{|V(G)|+1}{\alpha_{min}(G)+1}\right\rceil \right\}=\max\left\{k+1,\left\lceil \frac{k^3+k^2+k+1}{k^2+1}\right\rceil\right\}=k+1.\]
One can easily check that there is no equitable $(k+1)$-coloring of such graphs. On the other hand, it is easy to show that this graph is equitably $(k+2)$-colorable, hence $\chi_{=}(G)= k+2$.

\begin{figure}
    \centering
    \includegraphics{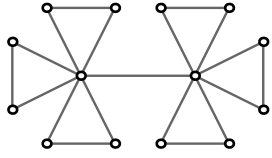}
    \caption{The $k=2$ case of the example.}
    \label{fig:k=2case}       
\end{figure}

\noindent The gap between $\chi_{=}(G)$ and $\max \left\{\omega(G),\left\lceil \frac{|V(G)|+1}{\alpha_{min}(G)+1}\right\rceil\right\}$ in the example above is one. This led us to the following conjecture, which is similar to the classical theorem of Vizing for graphs and the Goldberg conjecture for multigraphs:

\begin{conjecture}
\label{conj:gap1} For any block graph $G$, we have: 
\[\max \left\{\omega(G),\left\lceil \frac{|V(G)|+1}{\alpha_{min}(G)+1}\right\rceil\right\} \leq \chi_{=}(G) \leq 1+\max \left\{\omega(G),\left\lceil \frac{|V(G)|+1}{\alpha_{min}(G)+1}\right\rceil\right\}.\]
\end{conjecture}

\noindent We have confirmed the conjecture for all block graphs on at most 19 vertices, using a computer. Moreover, the conjecture is true for forests, i.e. for block graphs with $\omega(G)=2$ \cite{forests}.  Since the class of connected block graphs in which each cut vertex is on exactly two blocks is equivalent to line graphs of trees \cite{harary}, we have $ \chi_='(T)=\chi_=(G),$ for a tree $T$ and its line graph $G$. As trees are of Class 1 and $\Delta(T)=\omega(L(T))$ for a tree $T$, then we have $\chi_=(G)=\omega(G)$ for connected block graphs in which each cut vertex is on exactly two blocks. Thus our conjecture is true for such block graphs. Moreover, since an arbitrary graph $H$ is equitably edge $k$-colorable for every $k \geq \chi_='(H)$ \cite{furm:en_book}, then $\chi_=^*(G)=\chi_=(G)$ and the equitable chromatic spectrum of block graph $G$ in which each cut vertex is on exactly two blocks is gap-free. In this paper we prove the conjecture for well-covered block graphs, using the unusual tool of Ferrers matrix (Section \ref{sec:well}).  
Moreover, we prove our Vizing-Goldberg type conjecture for block graphs with small value of $\alpha_{\min}$ (Section \ref{small}) as well for some symmetric-like block graphs (Section \ref{symm}).

Finally, we would like to draw the reader attention +to the fact that the problem of block graphs coloring is closely related to very well known EFL conjecture formulated by  Erd\H{o}s, Faber, and Lov\'{a}sz in 1972 \cite{erdos}. They supposed that if $t$ complete graphs, each having exactly $t$ vertices, have the property that every pair of complete graphs has at most one shared vertex, then the union of the graphs can be colored with $t$ colors. Note that some block graphs are among those graphs that are affected by the conjecture. The only condition that they must fulfill is the size of all blocks is the same, and is equal to $t$. Note that the EFL conjecture implies that classical chromatic number of the graphs for which the conjecture holds is equal to their clique number. Hence, our results that say that some block graph $G$ can be equitably colored with $\omega(G)$ colors partialy confirm EFL conjecture - cf. Theorems \ref{thm:alpha=alphamin} and
\ref{thm:sym}.

\section{Well-covered block-graphs}\label{sec:well}

In this section we confirm Conjecture \ref{conj:gap1} for well-covered block graphs, i.e. for graphs fulfilling the condition  $\alpha_{min}(G)=\alpha(G)$. In this case we have the following chain of inequalities

\[\frac{|V(G)|+1}{\alpha_{min}(G)+1}\leq \frac{\alpha(G)\cdot \omega(G)+1}{\alpha(G)+1}\leq \omega(G).\] Hence the conjecture states that any such block graph must be equitably $\omega(G)$- or $(\omega(G)+1)$-colorable. In fact, we will show that well-covered block graphs are equitably $k$-colorable for all $k\geq \omega(G)$.
\subsection{Characterization}\label{character}
We start with a recursive characterization of well-covered block graphs. There is one basic class of such graphs, namely complete graphs. Now, we define the following operation: let $H$ be a block graph and let $v$ be a vertex of $H$. Add a clique $Q$ of size at least 2 to $H$ ($Q\cap H=\{v\}$), and add one pendant clique to each vertex of $Q$ except $v$. With this, all vertices of $Q$ will become cut-vertices and, in consequence, $Q$ becomes an internal clique. Let us note that the added pendant cliques may have different sizes. For further purpose we take the following notation. Let $v\in V(H)$. We add a clique $Q$ of size $s$, $s\geq 2$, with $s-1$ pendant cliques to graph $H$ by identifying one vertex of $Q$ with $v$. Let $V(Q)=\{v,v_1,\ldots,v_{s-1}\}$. Let $Q_i$ denote pendant clique added to vertex $v_i$, $|Q_i|=q_i$, $i \in [s-1]$. 
We order the pendant cliques due to their sizes in a non-increasing way creating a
vector $q=(q_1',\ldots, q_s')$, $q_1' \geq \cdots \geq q_s'$ (cf. Fig.~\ref{fig:ex}). Every $q_i'$, $i \in [s-1]$, corresponds to the size of one pendant clique. We assume that $q_s'=1$ and this term corresponds to the vertex $v$. 

\begin{theorem}
\label{thm:characAlpha=AlphaMin} Let $G$ be a well-covered block graph with a largest clique size $\omega(G)$. Then $G$ is either a basic graph, i.e. a clique of size at most $\omega(G)$, or $G$ can be obtained from a smaller well-covered block graph $H$ with $\omega(H)\leq \omega(G)$ by using the operation defined above, such that all involved cliques are of size at most $\omega(G)$.
\end{theorem}

\begin{proof} We prove this theorem in three steps. First of all, let us note that our basic class satisfies the condition $\alpha_{min}(G)=\alpha(G)$. Indeed, since  $G$ is a clique, then clearly $\alpha_{min}(G)=\alpha(G)=1$. 

Now, let us show that if we have a block graph $H$ with $\alpha_{min}(H)=\alpha(H)$ and we apply the operation given above (Case 1), then we will get a block graph $G$ with $\alpha_{min}(G)=\alpha(G)$. Assume that $v$ is the vertex of $H$ to whom we have added a clique $Q$ of size $s$ and have added $(s-1)$ pendant cliques to the vertices of $Q$ except $v$. Clearly, $\alpha(G)=\alpha(H)+s-1$. In order to complete the proof of this step, let us show that any vertex of $G$ lies in an independent set of $G$ of size $\alpha(H)+s-1$. Let $w$ be any vertex of $G$. If $w$ lies in $H$ or is one of simplicial vertices of pendant cliques adjacent to $Q$, then the statement is easy. Let us assume that $w$ is an internal vertex of $Q$. If it is $v$ then again the statement is easy. If it is a vertex of $Q$ different from $v$, then let $z$ be a vertex of $H$ adjacent to $v$. We can extend $z$ to an independent set of $H$ of size $\alpha(H)$, because $\alpha_{min}(H)=\alpha(H)$. Observe that $v$ does not belong to this independent set. Now add $w$ to it, and add $(s-2)$ simplicial vertices of pendant cliques (one simplicial vertex per pendant clique, except the one containing $w$). Observe that the resulting set is independent and its size is $\alpha(H)+1+(s-2)=\alpha(G)$. Thus, $\alpha_{min}(G)=\alpha(G)$.

In order to complete the proof of the theorem, let us show that if $G$ is a block graph with $\alpha_{min}(G)=\alpha(G)$, largest clique size $\omega(G)$ and outside from the basic class, then it can be obtained from a block graph $H$ with $\alpha_{min}(H)=\alpha(H)$ and the largest clique size $\omega(H)\leq \omega(G)$ by using our operation. First of all, let us observe that $G$ contains no two intersecting cliques, such that both of them contain simplicial vertices. If $G$ has two such cliques, then let $w$ be a vertex in their intersection. Extend $w$ to an independent set of $G$ of size $\alpha(G)$. Now, replace $w$ with two simplicial vertices of these cliques. Clearly, the resulting set is independent and its size is $\alpha(G)+1$, which is a contradiction.

Now let us remove all pendant cliques of $G$ except their cut-vertices. Because of the tree-structure of the block graph, we can always find a pendant clique $Q$ in the resulting graph. Observe that if we put back the removed cliques, then each vertex of $Q$ except the cut-vertex $v$ lies in a pendant clique of $G$. Let this clique be of size $s$. Observe that the vertices of $Q$ different from $v$ are intersecting with pendant cliques, one pendant clique per vertex different from $v$. Let $H$ be the graph obtained from $G$ by removing the $(s-1)$ pendant cliques intersecting $Q$ and the $(s-1)$ vertices of $Q$ different from $v$. Observe that $H$ is a block graph with $\omega(H)\leq \omega(G)$. Moreover, $\alpha(H)=\alpha(G)-(s-1)$. We claim that $\alpha_{min}(H)=\alpha(H)$. In order to see this, it suffices to show that any vertex $x$ of $H$ lies in an independent set of $H$ of size $\alpha(G)-(s-1)$. Since $\alpha_{min}(G)=\alpha(G)$, $x$ lies in an independent set of $G$ of size $\alpha(G)$. Without loss of generality, we can assume that this independent set takes simplicial vertices from pendant cliques intersecting $Q$ (one simplicial vertex per pendant clique). Observe that the number of these vertices is $s-1$. Hence, if we remove these $s-1$ vertices, we will obtain an independent set of $H$ having size $\alpha(G)-(s-1)=\alpha(H)$. Thus, $\alpha_{min}(H)=\alpha(H)$ and the proof is complete. \end{proof}

It worths to mention that our characterization completes the knowledge on characterization of different well-covered classes of graphs given in the literature (cf.\cite{prisner}).

\subsection{Some auxiliaries}

In the proof of the main result of this section we use a slightly modified concept of Ferrers matrix. 
Let $x=(x_1,\ldots,x_l)$ be any sequence with  $x_i \in \mathbb{Z}^{+}$, $i\in [l]$. It can be visualized by $l \times y$ matrix $M=(m_{i,j})$ of zeros and ones called  \emph{Ferrers matrix} for $x$ and defined by a sum of its row vector $r(M)=x$ and the properties:  (1) $y\geq l$, (2) if $m_{i,j}=0$ then $m_{i,t}=0$ for all $t \geq j$. In our concept \emph{modified Ferrers matrix} is also defined by row vector but, in addition, we have property $m_{i,l-i+1}=1$, while the second condition is modified to: if $m_{i,j}=0$ then $m_{i,t}=0$ for all $t \geq j$ excluding $t=l-i+1$ (cf. Fig.~(\ref{fig:ex2}a)). 
\begin{figure}[h]
    \centering
    \begin{tabular}{ll}
       a) & \hspace{2cm}b) \\
       $
   M=\left [
   \begin{array}{cccccc}
   1 & 1 & 0 & 1 & 0 & 0\\
   1 & 1 & 1 & 0 & 0 & 0\\
   1 & 1 & 0 & 0 & 0 & 0\\
   1 & 0 & 0 & 0 & 0 & 0\\
   \end{array}\right ]
   $  &  \hspace{2cm} 
   $
  M'=\left [
   \begin{array}{cccccc}
   0 & 0 & 1 & 1 & 0 & 1\\
   0 & 1 & 1 & 0 & 1 & 0\\
   1 & 1 & 0 & 0 & 0 & 0\\
   1 & 0 & 0 & 0 & 0 & 0\\
   \end{array}\right ]
   $
     \end{tabular}
   
    \caption{a) Modified Ferrers matrix $M$ for $q=(3,3,2,1)$ and $y=6$; b) An example of matrix $M'$ created due to the instructions given in the proof of Theorem \ref{mFn} corresponding to the matrix $M$ with vector $c(M')=p=(2,2,2,1,1,1)$.}
    \label{fig:ex2}
\end{figure}
Let $c(M)$ denote the vector of column sums in the matrix $M$. Given two arbitrary sequences $a=(a_1,\ldots,a_l)$ and $b=(b_1,\ldots b_y)$, such that $a_1+\cdots +a_l=b_1+\cdots+b_y$, we say that the vector $b$ is \emph{dominated} by the vector $a$ if $\sum_{i=1}^r a_i \geq \sum_{i=1}^r b_i$ for all positive integers $r$, where $a_i=0$ for all $i>l$ and $b_i=0$ for all $i>y$. We denote it by $a \unrhd b$.


\begin{theorem}
Let $a=(a_1,\ldots,a_l)$ and $b=(b_1,\ldots b_y)$ be two non-increasing sequences such that $a_1+\cdots +a_l=b_1+\cdots+b_y$.  Moreover, let $M$ be the modified Ferrers matrix for $a$. If $b$ is dominated by $c(M)$ then there exists $0$--$1$ matrix $M' = (m'_{i,j})$ such that $r(M')=a$ and $c(M')=b$, and  ${m'}_{i,l-i+1}=1$ for every $i \in [l]$.
\label{mFn}
\end{theorem}
\begin{proof}
The proof is based on the proof of sufficiency of Gale-Ryser theorem \cite{gale}. That proof is constructive. It starts from Ferrers matrix and requires finite number of applications of the following claim proved in \cite{gale}.
\begin{claim}
Given $l\times y$ matrix $M$ of zeros and ones such that $r(M) = a$, $c(M) \unrhd b$ and $c(M) \neq b$, we can find a $l\times y$ matrix $M'$ of zeros and ones such that $r(M') = a$, $c(M') \unrhd b$, and $\| c(M')-b \| < \| c(M)-b\|$ $($where $\| \cdot \|$ is the ordinary Euclidean norm$)$.
\label{clGale}
\end{claim}

Matrix $M'$ is constructed in the following way. First, we choose $i$ as minimal such that sum of $i$-th column of $M$ is greater than $q_i$ and choose $j$ as minimal such that sum of $j$-th column of $M$ is lower than $q_j$. Next, we can find a row index $h$ (there are at least two choices) such that $m_{h,i}=1$ and $m_{h,j}=0$. Matrix $M'$ is build from $M$ by swapping elements $m_{h,i}$ and $m_{h,j}$.
\\To prove Theorem~\ref{mFn} we start from modified Ferrers matrix and carefully apply Claim \ref{clGale}. In every step, when we choose row $h$ we take one such that $a_{h,i}$ is not of the form $a_{i',l-i'+1}$ (cf. Fig.~(\ref{fig:ex2}b)). 
\end{proof}

\begin{lemma}\label{claim1}
If  vector $b$ from Theorem \ref{mFn} fulfills the following conditions: $b_1 - b_y \leq 2$ while $b_i-b_y \leq 1$ for all $i \geq 2$, and $y \geq l$, then $b$ is dominated by $c(M)$.\label{cl1}
\end{lemma}
\begin{proof}
Let $c(M)=(c_1^M, \ldots,c_y^M)$. Since $c_1^M=l$ (due to the definition of modified Ferrers matrix), $y \geq l$, $\sum_{i=1}^l a_i=
\sum_{i=1}^y c_i^M=\sum_{i=1}^y b_i$
and vector $b$ is semi-balanced, i.e. fulfills the assumption of this lemma, we are sure that $c_1^M \geq b_1$. Thus, we have $\sum_{i=1}^r c_i^M \geq \sum_{i=1}^r b_i$ for all $r \leq s$ for some $s$, $1 < s \leq y$. To the contrary, let us assume that $b$ is not dominated by $c(M)$, i.e. $s<y$, $\sum_{i=1}^s c_i^M < \sum_{i=1}^s b_i$, and hence $c_s^M <b_s$. Since $\sum_{i=1}^l c_i^M=\sum_{i=1}^y b_i$, there must exist an index $x>s$ such that $c_x > b_x$. Since $b_s-b_x \leq 1$ and vector $c(M)$ is non-increasing, such situation is impossible. \end{proof}

\subsection{The main result}

We are ready to prove 
\begin{theorem}
\label{thm:alpha=alphamin} Let $G$ be a well-covered block graph. Then $G$ is equitably $k$-colorable for all $k\geq \omega(G)$.  
\end{theorem}
\begin{proof}
We use the recursive characterization of well-covered block graphs. It is clear that the theorem holds for the basic class of well-covered block graphs, i.e. for cliques. We assume that we have an equitable $k$-coloring of $H$ (coloring $f_c$) and we show that it is possible to extend it into the entire graph $G$ for every $k\geq \omega(G) \geq \omega(H)$.

The graph $G$ is given by a graph $H$ and the vector $q$ that describes the sizes of added cliques (cf. the characterization given in Subsection \ref{character}). Since an extended $k$-coloring of $G$ must be equitable, we can calculate how many times each color $i$ should be used in a coloring of $G-H+v$, $i\in [k]$, taking into account the equitable $k$-coloring of $H$. Let $p=(p_1,\ldots, p_k)$ be a non-increasing vector of cardinalities of color classes in a coloring of $G-H+v$. Moreover, let $col=
(c_1,\ldots, c_k)$ be a vector of colors such that the cardinality of color $c_i$ is equal to $p_i$ in a desirable coloring of $G-H+v$ (cf. Fig. \ref{fig:ex}). Note that there is at most one term in the sequence $p$ ($p_1$), such that there is a term $p_j$ fulfilling inequality $p_1-p_j=2$. For the remaining terms we have $|p_i-p_j|\leq 1$, and $i,j\geq 2$.
\begin{figure}[t]
    \centering
    \includegraphics[scale=0.8]{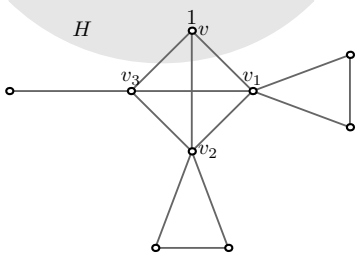}
    \caption{An example of graph $G-H+v$ with a counter of an equitable 6-coloring of $H$ given by the sequence $(10,10,10,10,9,9)$. Then $q=(3,3,2,1)$ and $p=(2,2,2,1,1,1)$ while the vector of the corresponding colors is, for example, as follows: $col=(1,5,6,2,3,4)$. }
    \label{fig:ex}
\end{figure}

On the further purpose we claim that the following property of the coloring of $H$ holds.
\begin{claim}
Graph $H$ may be recolored in such a way that the color assigned to $v$ is $1$ and term $p_1$ in vector $p$ represents color $1$, i.e. $c_1=1$. \label{clrec}
\end{claim}
\begin{proof} 
All we need is to prove that after the recoloring of graph $H$ color 1 is one of the colors with a largest cardinality in the desirable coloring of $G-H+v$.

Since the initial $k$-coloring of graph $H$ is equitable then the difference  between cardinalities of any two color classes is at most 1. 
Let us consider two cases: (1) the initial color of $v$, $f_c(v)$, is a color of a largest cardinality in the coloring of $H$. Then we do not have a term in the sequence $p$ such that the difference between this term and any other term from this sequence is 2. Moreover, the cardinality of color $f_c(v)$ is represented by a term of the sequence $p$ with the highest value. Thus, we may recolor graph $H$ in such a way that $v$ will obtain 1 and $q_1$ represents color 1; (2) the initial color $f_c(v)$ is a color of the smallest cardinality in the coloring of $H$. In this case we have two possible situations:
    \begin{itemize}
        \item we do not have a term in $p$ that exceeds any other term by exactly 2 - then the cardinality of color $f_c(v)$ is represented by a term of the sequence $p$ with the highest value,
        \item we have a term in $p$ that exceeds any other term by exactly 2, but it corresponds to color $f_c(v)$.  
    \end{itemize}
    In both cases, after recoloring $v$ to 1, $p_1$ represents color 1.
\end{proof}

To achieve the desirable equitable $k$-coloring of $G$ we start from the equitable $k$-coloring of $H$ such that vertex $v$ is colored with 1. Next, we build the modified Ferrers matrix for vector $q$. Observe that there is at most one term in the sequence $p$ ($p_1$), such that there is a term $p_j$ fulfilling inequality $p_1-p_j=2$, while for the remaining terms we have $|p_i-p_j|\leq 1$, $i,j\geq 2$. By Lemma \ref{claim1}, we may apply Theorem \ref{mFn} - vector $q$ corresponds to vector $a$ while vector $p$ corresponds to vector $b$. As a result, we get a $0$--$1$ matrix $M'$ such that $r(M')=q$ and $c(M')=p$ (cf. Fig.~\ref{fig:ex2}b)). Now, the desirable coloring of $G-H$ can be read from the matrix $M'$ in the following way. Let the pendant clique $Q_i$ of size $q_i$ be represented by the term $q_j'$ in sequence $q$. Then, $f_c(v_i):=c_{l-j+1}$, and the remaining vertices in $Q_i$ are colored with all colors $c_u$ such that $m'_{j,u}=1$ and $u \neq l-j+1$ (cf. Fig.~\ref{fig:ex42}). 
\end{proof}
\begin{figure}[t]
    \centering
    \includegraphics[scale=0.8]{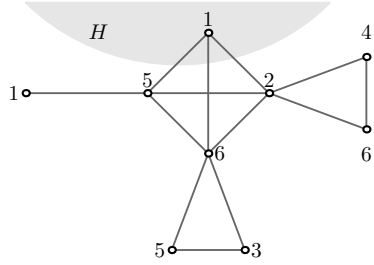}
    \caption{An equitable coloring of $G-H+v$ read from matrix $M'$ given in Figure \ref{fig:ex2}b). }
    \label{fig:ex42}
\end{figure}

\begin{corollary}
The equitable chromatic spectrum of well-covered block graphs is gap-free.
\end{corollary}

\section{Symmetric-like block graphs}\label{symm}

In this section, we present some results that deal with symmetric-like block graphs. A \emph{symmetric block graph} $B(n,k)$ is a block graph whose all blocks are cliques of size $n$, each cut vertex belongs to exactly $k$ blocks, and the eccentricity of simplicial vertices is same. Inspired by this definition, we define some other subclasses of block graphs and confirm Conjecture \ref{conj:gap1} for them.

\subsection{Class $B_l(n,k)$, $k\geq 3$}

We consider a subclass of block graphs that is defined recursively. Let $B_1(n,k)$ be a clique of size $n$. A graph $B_l(n,k)$, $l\geq 2$, is built from $B_{l-1}(n,k)$ by adding to each simplicial vertex $k-1$ pendant cliques, each one of size $n$ (cf. Fig. \ref{fig:l3} with $n=k=l=3$). 

\begin{figure}
    \centering
    \includegraphics{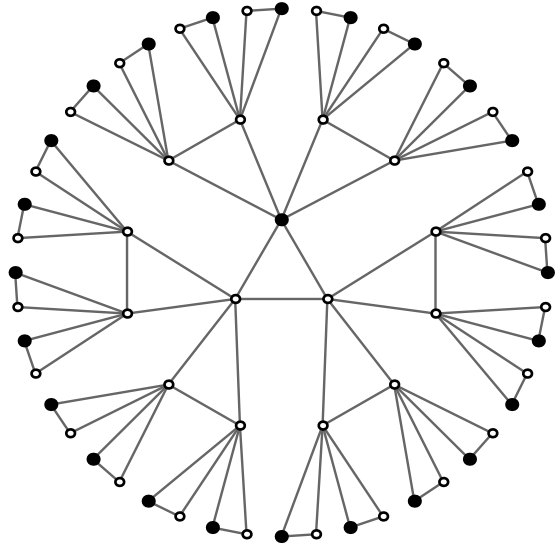}
    \caption{Graph $B_3(3,3)$ with an exemplary independent set of the largest size depicted by black vertices.}
    \label{fig:l3}
\end{figure}

Note that $$|V(B_l(n,k))|=\Sigma _{x=0}^{l-1} n(k-1)^x(n-1)^x.$$
It is easy to observe that the largest independent set in a graph $B_l(n,k)$ is formed by simplicial vertices - one such a vertex from each pendant clique - and appropriate cut vertices from every second ''level'' (cf. Fig.~\ref{fig:l3}). Thus, we have 
$$\alpha(B_l(n,k))=\left \{ 
\begin{array}{ll}
    1 & \text{if }l=1,\\
   1+\Sigma_{x=1}^{l-2} n (k-1)^{x+1}(n-1)^x  & \text{if } l\geq 3 \text{ and } l \text{ is odd,} \\
\Sigma_{x=0}^{l-2} n (k-1)^{x+1}(n-1)^x     & \text{if }l \text{ is even.}
\end{array}
\right.$$

\begin{lemma}
For every graph $G$ being a $B_l(n,k)$, $l\geq 2$, we have
$$\alpha(G)-\alpha_{\min}(G)=k-2.$$
\end{lemma}
\begin{proof}
It is enough to observe that an independent set of the smallest size is designated by a cut vertex from any pendant clique.
\end{proof}

\begin{lemma}
For every graph $G$ being a $B_l(n,k)$, $l\geq 1$, we have $$\max \left\{\omega(G), \left \lceil \frac{|V(G)|+1}{\alpha_{min}(G)+1}\right \rceil\right\}=\omega(G).$$
\end{lemma}
\begin{proof}
Observe that in the construction of block graph $B_l(n,k)$ we always work with cliques of size $n$, hence for such graphs, we have $\omega(G)=n$. Thus, in order to prove the above equality, it suffices to show that
\[\frac{|V(G)|+1}{\alpha_{min}(G)+1}\leq n.\]
We split the proof of this into three cases. The first case is when $l=1$. In this case, we have a clique of size $n$, and the inequality is trivially true. Now, we prove when $l$ is even. Hence $l\geq 2$. As we stated in the introduction, we can assume that $n\geq 3$ and $k\geq 3$.

In this case, observe that 
\begin{equation}\label{eq:AUX}
    2nk\leq n+k+(k-1)n^2.
\end{equation}
In order to see this, observe that for case $n=3$, it follows from the fact that $k\geq 3$. On the other hand, if $n\geq 4$, then
\[\frac{2k}{k-1}=2+\frac{2}{k-1}\leq 4\leq n,\]
hence
\[2nk\leq (k-1)n^2,\]
which clearly implies the inequality (\ref{eq:AUX}). Now, by elementary manipulations, it can be shown that inequality (\ref{eq:AUX}) implies
\[n^2(k-1)+n(n-1)(k-1)(k-3)+n\leq (k-1)^2(n-1)^2+(k-1)^2(n-1)+k.\]
Since $l\geq 2$, we have
\[n^2(k-1)+n(n-1)(k-1)(k-3)+n\leq (k-1)^l(n-1)^l+(k-1)^l(n-1)^{l-1}+k.\]
Again, simple manipulations imply that this inequality is the same as
\[(k-1)(n-1)-1+n(k-1)^l(n-1)^l-n\leq n\cdot [(3-k)((k-1)(n-1)-1)+n(k-1)^l(n-1)^{l-1}-n(k-1)].\]
The latter implies
\[\frac{|V(G)|+1}{\alpha_{min}(G)+1}=\frac{1+n\cdot \frac{(k-1)^l(n-1)^l-1}{(k-1)(n-1)-1}}{3-k+n(k-1)\cdot \frac{(k-1)^{l-1}(n-1)^{l-1}-1}{(k-1)(n-1)-1}}\leq n,\]
which is the bound we were trying to prove. This just follows from the formula for the sum of members of a geometric progression and the above mentioned formulas for $|V|$, $\alpha$ and $\alpha_{\min}$. 

Finally, we are left with the case $l\geq 3$ and $l$ is odd. This case can be done in a way similar to the case $l$ even. 
\end{proof}

To confirm Conjecture \ref{conj:gap1} for $B_l(n,k)$ graphs we prove the following 
\begin{theorem}
Every block graph $G=B_l(n,k)$, $l\geq 1$, has an equitable $\omega(G)$-coloring.\label{thm:sym}
\end{theorem}
\begin{proof}
We start from $\omega(G)$-coloring of $B_1(n,k)$. Then we extend the coloring for vertices from next levels, for $B_2(n,k)$, $B_3(n,k)$, $\ldots$, $B_l(n,k)$. The $\omega(G)$-coloring is determined (due to permutations of colors) and it is easy to show that the coloring is equitable, even strongly equitable. This means, each color is used exactly the same number of times.
\end{proof}

\begin{corollary}
There is infinite family of block graphs for which we have arbitrarily large difference $\alpha(G)-\alpha_{\min}(G)$ and such graphs demand $\omega(G)$ colors to be equitably colored.
\end{corollary}

\noindent Remember that we have given an example of infinite family of block graphs defined by the condition: $\alpha(G)=\alpha_{\min}(G)$ also demanding $\omega(G)$ colors to be equitably colored, in Section \ref{sec:well}.

In conclusion, we have
\begin{corollary}
Conjecture \ref{conj:gap1} holds for $B_l(n,k)$ graphs.
\end{corollary}

Taking into account our consideration on the  equitable chromatic spectrum for well-covered block graphs in Section \ref{sec:well}, we may state an open question for $B_l(n,k)$ graphs.

\begin{question}
Is it true that graph $G$, $G=B_l(n,k)$, is equitably $k$-colorable for every $k \geq \omega(G)$? In other words: Is the equitable chromatic spectrum gap-free for $B_l(n,k)$ graphs?
\end{question}

\subsection{Class $\mathcal{B}(3,\leq 3)$}

In the case of symmetric block graphs and block graphs $B_l(n,k)$ the number of cliques to which a cut-vertex belongs was determined. Now, we weaken this condition and define the subclass of block graphs, denoted by $\mathcal{B}(n, \leq k)$, as the subclass of block graphs such that every block is a clique of size $n$ and every cut vertex belongs to at most $k$ blocks.
In this subsection, we confirm Conjecture \ref{conj:gap1} for graphs from the class $\mathcal{B}(3,\leq 3)$. Observe that $\omega(G)=3$ for $G\in \mathcal{B}(3,\leq 3)$. We show that each such graph $G$ admits an equitable $4$-coloring. We need some auxiliary definitions.

\begin{definition}
A graph $G=(V,E)$ (not necessarily block graph) with designated vertex $v \in V(G)$ is of type:
\begin{itemize}
    \item[\emph{(T1)}] if $G$ is equitably $4$-colorable with colors: $A$, $B$, $C$, and $D$ that occur $m+1, m, m,$ and $m$ times, respectively, and vertex $v$ is colored with $A$;
    \item[\emph{(T2)}] if $G$ is equitably $4$-colorable with colors: $A$, $B$, $C$, and $D$ that occur $m+1, m+1, m+1,$ and $m$ times, respectively, and vertex $v$ is colored with $A$.
\end{itemize}

\end{definition}

\noindent Observe that every graph from the class $\mathcal{B}(n,\leq k)$ includes a vertex of degree $n-1$. 

\begin{lemma}
A connected graph $G$, $G\in \mathcal{B}(3,\leq 3)$, with any vertex $v \in V(G)$ of degree $2$ is of type \emph{(T1)} or \emph{(T2)}.
\end{lemma}

\begin{proof}
We will prove this lemma by induction on the order of graph $G$. First, note that $K_1$ (one isolated vertex $v$) is of type (T1) with $v$, and the triangle, with any of its vertices, is of type (T2).

Now, consider a graph $G$ on $n>3$ vertices, $G\in\mathcal{B}(3,\leq 3)$, and assume that lemma is true for graphs from $\mathcal{B}(3,\leq 3)$ of order smaller than $n$. Let $v$ be a vertex in $G$ of degree 2. It belongs only to one block $K_3$ and it has two neighbors, let us say $a$ and $b$. The vertex $a$ (resp. $b$) belongs to at most 3 blocks, thus it may be included in two other triangles. 
Let us consider the connected subgraph of $G-\{v,b\}$ that contains $a$. We name it by $G_a$. Now, let us consider $G_a-a$. It consists of at most two connected components: $G_{a_1}$ and $G_{a_2}$. Let $G_1$ (resp. $G_2$) be $G_{a_1}$ (resp. $G_{a_2}$) after adding vertex $a$  and restoring the original edges between $a$ and vertices of $G_{a_1}$ (resp. $G_{a_2})$. If $G_a-a$ consists of less than two components, then an appropriate graph $G_1$ or/and $G_2$ is a single vertex $a$ (cf. Fig.~\ref{fig::grG}). Similarly, we define subgraphs $G_3$ and $G_4$ for vertex $b$ (cf. Fig.~\ref{fig::grG}). 
Every graph $G_i$ is $K_1$ or it belongs to $\mathcal{B}(3,\leq 3)$ with $|V(G_i)|<n$, $1 \leq i \leq 4$. In both cases, every graph $G_i$ is of type (T1) or (T2). 

We have six possibilities, up to an isomorphism:
\begin{itemize}
    \item 
        $G_1$ with vertex $a$ is of type (T1), 
        $G_2$ with vertex $a$ is of type (T1), 
        $G_3$ with vertex $b$ is of type (T1), 
        $G_4$ with vertex $b$ is of type (T1). 
        For colors in $G_1$ and $G_2$ we use the permutation $(AB)(C)(D)$. For colors in $G_3$ and $G_4$ we use $(AC)(B)(D)$. After we color vertex $v$ in color $A$, the graph $G$ is of type (T2).
    \item 
        $G_1$ with vertex $a$ is of type (T1), 
        $G_2$ with vertex $a$ is of type (T1), 
        $G_3$ with vertex $b$ is of type (T1), 
        $G_4$ with vertex $b$ is of type (T2). 
        For colors in $G_1$ and $G_2$ we use the permutation $(AD)(B)(C)$. For colors in $G_3$ and $G_4$ we use $(AB)(C)(D)$. After we color vertex $v$ in color $A$, the graph $G$ is of type (T1).
    \item 
        $G_1$ with vertex $a$ is of type (T1), 
        $G_2$ with vertex $a$ is of type (T1), 
        $G_3$ with vertex $b$ is of type (T2), 
        $G_4$ with vertex $b$ is of type (T2). 
        For colors in $G_1$ and in $G_2$ we use the permutation $(AB)(C)(D)$. For colors in $G_3$ we use $(ADB)(C)$,
        and in $G_4$ we use $(AD)(B)(C)$. After we color vertex $v$ in color $A$, the graph $G$ is of type (T2).
    \item 
        $G_1$ with vertex $a$ is of type (T1), 
        $G_2$ with vertex $a$ is of type (T2), 
        $G_3$ with vertex $b$ is of type (T1), 
        $G_4$ with vertex $b$ is of type (T2). 
        For colors in $G_1$ we use the permutation $(AB)(C)(D)$
        and in $G_2$ we use the permutation $(ABD)(C)$. For colors in $G_3$ and $G_4$ we use $(AC)(B)(D)$. After we color vertex $v$ in color $A$, the graph $G$ is of type (T2).
    \item 
        $G_1$ with vertex $a$ is of type (T1), 
        $G_2$ with vertex $a$ is of type (T2), 
        $G_3$ with vertex $b$ is of type (T2), 
        $G_4$ with vertex $b$ is of type (T2). 
        For colors in $G_1$ and $G_2$ we use the permutation $(AC)(B)(D)$. For colors in $G_3$ we use $(ABD)(C)$, and in $G_4$ we use $(ABDC)$. After we color vertex $v$ in color $A$, the graph $G$ is of type (T1).
    \item 
        $G_1$ with vertex $a$ is of type (T2), 
        $G_2$ with vertex $a$ is of type (T2), 
        $G_3$ with vertex $b$ is of type (T2), 
        $G_4$ with vertex $b$ is of type (T2). 
        For colors in $G_1$ and $G_2$ we use the permutation $(AB)(C)(D)$. For colors in $G_3$ and $G_4$ we use $(ACD)(B)$. After we color vertex $v$ in color $A$, the graph $G$ is of type (T2).
\end{itemize}

\begin{figure}
    \centering
    \includegraphics{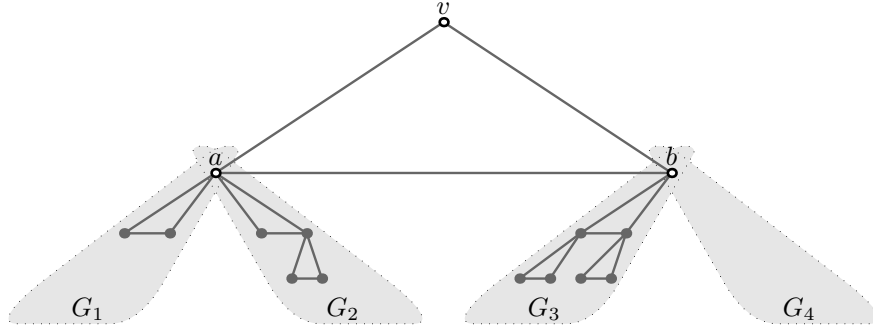}
    \caption{An example of graph $G \in \mathcal{B}(3,\leq 3)$}
    \label{fig::grG}
\end{figure}
\end{proof}

\begin{corollary}
Conjecture~\ref{conj:gap1} holds for every $G\in \mathcal{B}(3,\le3)$.
\end{corollary}

Note that this result cannot be extended to large values of $k$. In other words, it is not true that each $\mathcal{B}(k,\le k)$ is equitably $(k+1)$-colorable (cf. Fig. \ref{fig:4444}).

\begin{figure}[htb]
    \centering
    \includegraphics{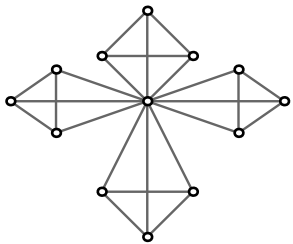}
    \caption{An example of a block graph from class $\mathcal{B}(4,\leq 4)$ demanding more than 5 colors to be equitably colored. Here $\chi_=(G)=\lceil \frac{|V(G)|+1}{\alpha_{\min}(G)+1}\rceil=7. $}
    \label{fig:4444}
\end{figure}

On Figure \ref{fig:BlockGraphMax+1}, we give an example of a block graph from $\mathcal{B}(3,\leq 3)$ that attains the upper bound in Conjecture~\ref{conj:gap1}. In this example, $|V|=11$, $\omega(G)=3$ and $\alpha_{\min}(G)=3$. Thus,
\[\max\left\{\omega(G), \left\lceil \frac{|V|+1}{\alpha_{\min}(G)+1}\right\rceil\right\}=\max\left\{3,\left\lceil \frac{12}{4}\right\rceil\right\}=3.\]
It can be easily seen that the only way to partition 11 vertices of this graph into three almost equal size is to have 4, 4 and 3 vertices in each set. Using this observation, it can be shown (with a reasoning similar to the graph from Figure \ref{fig:k=2case}) that this graph is not equitably 3-colorable. On the other hand, it is equitably 4-colorable (see the coloring on Figure \ref{fig:BlockGraphMax+1}).
\begin{figure}[h]
    \centering
    \includegraphics{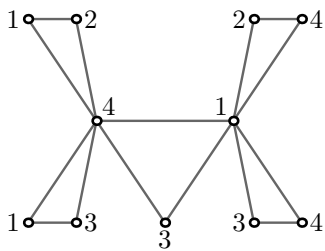}
    \caption{A graph from the class $\mathcal{B}(3,\leq 3)$ attaining the upper bound in Conjecture~\ref{conj:gap1}.}
    \label{fig:BlockGraphMax+1}       
\end{figure}

\section{Block-graphs with small $\alpha_{\min}$}\label{small}

In this section, we prove our conjecture for block-graphs with small value of $\alpha_{\min}$. We start with $\alpha_{\min}(G)=1$. Note that any graph $G$ with $\alpha_{\min}(G)=1$ has one vertex $v$ adjacent to any other vertex of $G-v$. In the case of block graphs, a graph consists of cliques $Q_1,\ldots, Q_t$, $t \geq 1$, sharing one common vertex.
Let $n_i$ denote the size of the clique $Q_i$, $1 \leq i \leq t$. Assume that cliques are ordered according to their sizes, in a non-decreasing way, i.e. $n_1\leq n_2 \leq \cdots \leq n_t$.
We denote such a block graph by $B_{n_1,n_2,\ldots,n_t}$, $t \geq 1$. 
Let us think about
equitable coloring of $G$, $G=B_{n_1,n_2,\ldots,n_t}$, with the smallest number of colors. 
The color assigned to the universal vertex $v$ of $G$ cannot be used to color any other vertex in $G$. 
Since coloring must be equitable, every other color in this coloring can be used at most twice. 
Thus, we need to partition the vertex set of $G$ into minimum number of color classes of size at most 2. 
This problem is equivalent to finding a maximum matching in the complement of $G-v$. Note that $\overline{G-v}$ is a complete multipartite graph $K_{n_1-1,\ldots,n_t-1}$. The maximum matching problem in complete multipartite graphs was considered in \cite{multi}. Due to the results given in the paper, we have the following.

\begin{theorem}[Thm. 1 in \cite{multi}, original notations]
Let $K_{m_1,m_2,\ldots,m_n}$ be a complete multipartite graph with $m_i$ vertices in the $i$th part, labeled so that $m_1 \leq m_2 \leq \ldots \leq m_n$. If $m_n \geq m_1+m_2+\cdots+m_{n-1}$,
then:
\begin{enumerate}
    \item[(i)] the number of edges in any maximum matching is $M = m_1+m_2+\cdots+m_{n-1}$;
    \item[(ii)]  a maximum matching can be obtained by connecting all vertices in the parts with $m_1, m_2, \ldots,$ $m_{n-1}$ vertices to vertices in the part with $m_n$ vertices.
\end{enumerate}\label{thm1a}
\end{theorem}

Considering our block graph $G$, $G=B_{n_1,n_2,\ldots,n_t}$, such that $\overline{G-v}$ fulfills the assumption of Theorem \ref{thm1a}, which implies that in an optimal equitable coloring of $G$ we have $M$, $M=(n_1-1)+\cdots + (n_{t-1}-1)$ color classes of size 2 and $(n_{t-1})-M+1$ color classes of size 1. Thus, we have $n_t$ color classes in total. Hence, $\chi_=(G)=\omega(G)$ in this case.

\begin{theorem}[Thm. 3 in \cite{multi}, original notation]
Given any complete multipartite graph, $K_{m_1,m_2,\ldots,m_n}$ , where $m_1 \leq
m_2 \leq \ldots \leq m_n$ and $m_n < m_1 + m_2 + \cdots + m_{n-1}$, the number of edges in any maximum matching is $M = \lfloor (m_1+\cdots+m_n) /2 \rfloor$.\label{thm1b}
\end{theorem}

Note, that considering our block graph $G$, $G=B_{n_1,n_2,\ldots,n_t}$, such that $\overline{G-v}$ fulfills the assumption of Theorem \ref{thm1b}, which implies that in an optimal equitable coloring of $G$ we have $\lfloor (|V(G)|-1)/2\rfloor$ color classes of size two and 1 or 2 (depending on parity of $|V(G)|-1$) color classes of size one. Thus, we have $\lceil (|V(G)|-1)/2 \rceil+1$ color classes in total. Hence, $\chi_=(G)=\lceil (|V(G)|+1)/2 \rceil$.

\begin{corollary}
Conjecture \ref{conj:gap1} is true for block graphs with $\alpha_{\min}(G)=1$.
\end{corollary}

Next, we give a characterization of connected block graphs that satisfy the condition  $\alpha_{\min}(G)=2$. First, we observe some properties of such graphs. 

\begin{observation}
For each $x\in V(G)$ with $\alpha(G,x)=2$ there exists a maximal independent set $I_{\min}=\{x,v\}$ such that at least one out of $x$ and $v$ is a cut vertex. \label{obscut}
\end{observation}
\begin{proof}
If $x$ is a cut vertex, we are done. So let us assume that $x$ is a simplicial vertex that belongs to a clique $Q^x$. Since $\alpha(G,x)=2$, the second vertex, let it be $y$, from $I_{\min}$ including $x$, must belong to another clique, $y \in Q^y$. Moreover,   $Q^x \cap Q^y = \emptyset$. Otherwise the maximal independent set $\{x,y\}$ could be replaced by the vertex from $Q^x \cap Q^y$, a contradiction with $\alpha_{\min}(G)=2$. If $y$ is a cut vertex, we are done. So, let us assume that $x$ and $y$ are simplicial vertices.
But then one of them, let us say $x$, can be replaced in $I_{\min}$ by the cut vertex belonging to $Q^x$ and lying on the path $x-y$, what finishes the proof.
\end{proof}

\begin{observation}
Let $I_{\min}=\{x,y\}$ be a maximal independent set in $G$, $x \in Q^x$, $y\in Q^y$, and let $x$ be a cut vertex. Then only one of the following conditions hold: (1) There is exactly one clique, let us say $Q^z$, on the path $x-y$, different from $Q^x$ and $Q^y$. (2) There are exactly two cliques on the path $x-y$, different from $Q^x$ and $Q^y$, but the clique intersecting non empty clique $Q^y$ is $K_2$. 

\noindent Any other situation according the cliques  on the path $x-y$ is impossible.
\label{obsk2}
\end{observation}
\begin{proof}
It is easy to observe that there must exist at least one clique on the path $x-y$ different from $Q^x$ and $Q^y$. We have already given the reason in the proof of Observation \ref{obscut}. 
Let us consider the situation when we have exactly one clique, different from $Q^x$ and $Q^y$, on the path $x-y$. If the clique were of size greater than two, we could add any of its simplicial vertices to the independent set $I_{\min}$ - contradiction with $I_{\min}$ is maximal.
If we had two or more cliques on the path $x-y$, we also would increase the size of $I_{\min}$.
\end{proof}
\begin{observation}
Let $I_{\min}=\{x,y\}$ be a maximal independent set in $G$, $x \in Q^x$, $y\in Q^y$, and let $x$ be a cut vertex. There is no cut vertex in $Q^x$ other than $x$.
\end{observation}
\begin{proof}
If there were a clique $Q^z$ such that $Q^z \cap Q^x \neq \emptyset$ and $Q^z$ does not lie on the path $x-y$, then we could add any vertex from $Q^z \backslash Q^x$ to $I_{\min}$ - a contradiction with its maximality.
\end{proof}

Let us ask whether there may exist any clique adjacent to $Q^y$ not lying on the path $x-y$. It turns out that the only feasible situation for the positive answer is when we have only one clique on the path $x-y$ different from $Q^x$ and $Q^y$ and $Q^y$ is isomorphic to $K_2$. Then we may add one clique of any size to $Q^y$ by identifying any vertex of the added clique with the simplicial vertex of $Q^y$. But it results in the situation described in Observation \ref{obsk2} point 2. So, we may assume that $Q^x$, as well as $Q^y$ has exactly one cut vertex.

It remains to consider whether there may exist any other cliques outside the path $x-y$. If we add any other cliques to the cut vertex $x$, the independent set including $x$ does not change, so such situation is feasible. Observe that we cannot add any additional clique to any other cut vertex in $G$ without increasing the value of $\alpha_{\min}(G)$. 

Summarizing, a connected block graph $G$ with $\alpha_{\min}(G)=2$ is characterized in the following way. We start from a star of cliques $B_{n_1,n_2,\ldots,n_t}$, $t\geq 2$, and then one of the following two operations is allowed:
\begin{enumerate}
    \item we add a clique $Q_0$ of size $n_0$, $n_0 \geq 2$, to any clique of  $B_{n_1,n_2,\ldots,n_t}$, let us say to $Q_l$ of size $n_l$, $l \in [t]$, by identifying any vertex of $Q_0$ with any simplicial vertex of $Q_l$; we obtain a graph $B_{n_l-n_0, n_1,\ldots, n_{l-1}, n_{l+1}, \ldots, n_t}$ (cf. Fig. \ref{fig:ex_alpha2}a));
    \item we add $K_2$ to any clique of $B_{n_1,n_2,\ldots,n_t}$, let us say to a clique $Q_l$ of size $n_l$, $l \in [t]$, by identifying a vertex of added $K_2$ with any simplicial vertex of $Q_l$, and then we add a clique $Q_0$ of size $n_0$, $n_0 \geq 2$, to pendant vertex of added earlier $K_2$, by identifying any vertex of $Q_0$ with the pendant vertex; we obtain a graph $B_{n_l-2-n_0, n_1,\ldots, n_{l-1}, n_{l+1}, \ldots, n_t}$ (cf. Fig. \ref{fig:ex_alpha2}b)).
\end{enumerate}

\begin{figure}[t]
    \centering
    a)
    \includegraphics[scale=0.7]{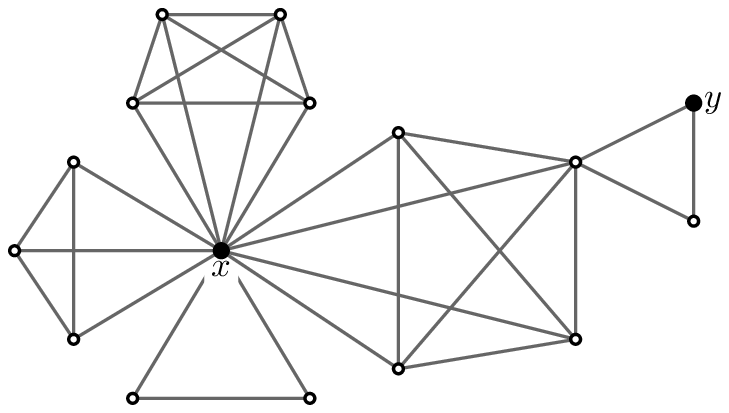} b)
    \includegraphics[scale=0.7]{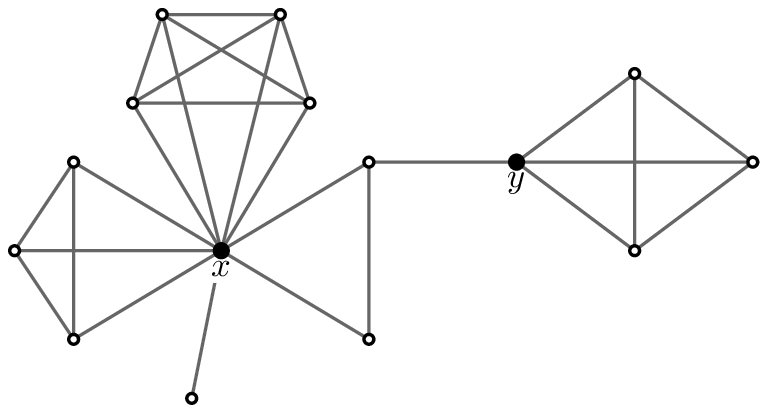}
    \caption{An example of two connected block graphs with $\alpha_{\min}=2$: a) $B_{5-3,5,4,3}$, b) $B_{3-2-4,5,4,2}$. A maximal independent set of size 2 is depicted by black vertices.}
    \label{fig:ex_alpha2}
\end{figure}

Observe that this characterization theorem works only when $G$ is a connected graph. A disconnected example that is not of the type described above would be a graph with two components such that each component is a clique.
Now, we use the characterization in order to show that all connected block-graphs with $\alpha_{\min}=2$ can be equitably colored with
$
k=\max\left\{\omega(G), \left\lceil \frac{|V(G)|+1}{\alpha_{min}(G)+1}\right\rceil\right\}$
colors. 

\begin{proposition}
    Conjecture \ref{conj:gap1} is true for every connected block graph $G$ with $\alpha_{\min}(G)=2$.
\end{proposition}
\begin{proof}
We consider two cases, dependly on the structure of $G$.
\begin{description} 
\item[Case 1] $G$ is of type $B_{n_l-n_0, n_1,\ldots, n_{l-1}, n_{l+1}, \ldots, n_t}$.

Let $x$ be the common vertex of cliques $Q_1,\ldots,Q_t$, while $y$ be any simplicial vertex of $Q_0$. It is easy to see that $\{x,y\}$ is a largest independent set including the vertex $x$. This set forms the first color class, let us say $V_1$. Observe that $G-I$ consists of cliques $Q_1-x, \ldots, Q_{l-1}-x, Q_{l+1}-x, \ldots, Q_t-x$ of sizes $n_1-1,\ldots,n_{l-1}-1, n_{l+1}-1,\ldots,n_t-1$, respectively, and a graph $B'$ being composed from $Q_l-x$ and $Q_0-y$ that share one vertex. First, we color $G-I$ starting from vertices of $B'$
    with colors $2, 3,\ldots, n_1,n_1+1,\ldots, n_1+n_0-1$, if $n_1+n_0 \leq k+1$, or with colors $2,3,\ldots, k, 2,3, \ldots, n_1+n_0-k$, otherwise. In the further case, we simply use colors $\{2,3,\ldots,k\}$ periodically. Note, that since $k \geq \omega(G)$ and $n_1+n_0 \leq 2\omega(G)-1$, the coloring of $B'$ is proper. Assume that the vertex of $Q_0-y$ colored as the last one, in the coloring of $B'$ described above, is colored with color $k'$.
    Next, we color vertices of cliques $Q_1-x, \ldots, Q_{l-1}-x, Q_{l+1}-x, \ldots, Q_t-x$, clique by clique, with colors $2,\ldots,k$ periodically, starting from color 2, if $k'=k$, or with color $k'+1$, otherwise. 
    Again, since $k \geq \omega(G)$, the coloring of $G-I$ is proper, and together with class $V_1$, we obtain a proper $k$-coloring of $G$. 
    \begin{claim}
    The coloring of $G$ is equitable.
    \end{claim}
    \begin{proof}
    Since we use colors $2,\ldots,k$ periodically, certainly the coloring of $G-I$ is equitable (taking into account only colors from $\{2,\ldots,k\}$). 
    Next, observe that each color in the coloring of $G-I$ is used at most three times. On the contrary, we may suppose that at least one color is used four times. Then, $|V(G-I)|\geq 3(k-1)+1$, what implies $|V(G)| \geq 3k$, and this in turn is a contradiction with $k \geq (|V|+1)/3$.
    \\On the other side, since $k \geq \omega(G)$, each color is used at least once.
    Thus, the equitable coloring of $G-I$ together with the color class $V_1$ of cardinality 2 forms the equitable coloring of $G$.
    \end{proof} 
\item[Case 2] $G$ is of type $B_{n_l-2-n_0, n_1,\ldots, n_{l-1}, n_{l+1}, \ldots, n_t}$

Let $x$ be the common vertex of cliques $Q_1,\ldots,Q_t$, while $y$ be now the cut vertex of $Q_0$. It is easy to see that $\{x,y\}$ is a largest independent set including the vertex $x$. This set forms the first color class, let us say $V_1$. Observe that this time $G-I$ is a union of cliques. We color vertices of cliques $Q_1-x, \ldots, Q_t-x$, and $Q_0-y$, clique by clique, with colors $2,\ldots,k$ periodically, starting from color 2.
Since $k \geq \omega(G)$, the coloring of $G-I$ is proper, and together with class $V_1$, we obtain a proper $k$-coloring of $G$. Similarly to the previous case, we can prove that the coloring of $G$ is equitable.
\end{description}
\end{proof}

\section{Conclusion and future work}
In this paper, we considered equitable colorings of block graphs, which form a proper subclass of chordal graphs. The problem is interesting in this class as it can be shown that the problem of equitable coloring of chordal graphs is equivalent to equitable edge coloring of hypertrees, an interesting and non-trivial subclass of hypergraphs. Since the problem is NP-hard in the class of chordal graphs, a natural approach is to consider its restriction to various subclasses of chordal graphs. Block graphs are such examples. 

In the paper, we considered Conjecture \ref{conj:gap1}, which was offering a bound for equitable chromatic number, such that the difference between the upper and lower bounds is at most one. Moreover, both of the bounds are computable in polynomial time. Thus, in some sense, the situation is similar to the chromatic index of graphs, where for simple graphs there is the classical theorem of Vizing and for multigraphs there is the Goldberg conjecture, where a similar gap-one bound is offered for this parameter in the class of all multigraphs (see \cite{GoldbergProof}, where a proof of the latter conjecture is announced).

In the paper, we verified our conjecture for various subclasses of block graphs. Moreover, we gave various examples of block graphs, for which both lower and upper bounds of Conjecture \ref{conj:gap1} are tight. Usually, when one considers equitable colorings, there are two parameters that one takes into account: the smallest number of colors in an equitable coloring of a graph (equitable chromatic number, $\chi_=(G)$), and the smallest $k$, such that the graph admits an equitable $t$-coloring for any $t\geq k$ (equitable chromatic threshold, $\chi_=^*(G)$). As complete bipartite graphs show, these two parameters are not always the same. However, the results obtained in this paper confirm our belief that these two parameters have to be the same in the class of block graphs, though we do not have a complete proof of this statement.

One may wonder whether the statement of Conjecture \ref{conj:gap1} can be extended to arbitrary graphs. In order to see that this extension cannot be true, consider the complete tripartite graph $G=K_{3,5,7}$. Recall that this graph can be obtained from three disjoint independent sets of size 3,5 and 7, respectively, by joining any two vertices lying in different independent sets with an edge. Observe that $|V|=15$, $\omega(G)=3$ and $\alpha_{\min}(G)=3$. Hence the upper bound from Conjecture \ref{conj:gap1} is five. However, it is not hard to see that $\chi_=(G)=6$.
What is even more interesting that the conjecture cannot be extended even into the whole class of chordal graphs. 
\begin{figure}
    \centering
    \includegraphics[scale=1]{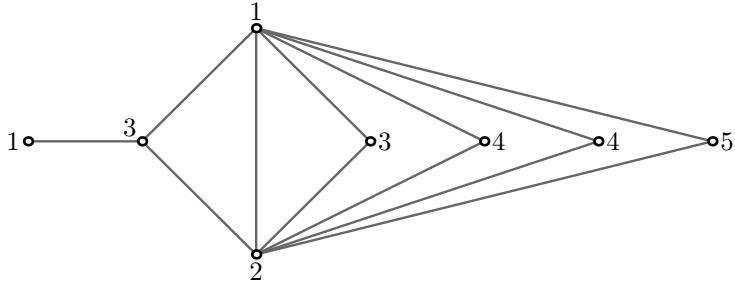}
    \caption{An example of chordal graph with its equitable coloring.}
    \label{fig:ex:chord}
\end{figure}
For the chordal graph $G$ from Figure  \ref{fig:ex:chord} we have: $|V|=8$, $\omega(G)=3$ and $\alpha_{\min}(G)=2$. Thus, the upper bound from Conjecture \ref{conj:gap1}: $1+\max \left\{\omega(G), \frac{|V(G)|+1}{\alpha_{min}(G)+1}\right\}$ is equal to 4, while we need 5 colors to color the graph equitably. Note that our exemplary graph has two vertices $v_1, v_2$, the ones of the highest degree (cf. Fig. ~\ref{ex_gen}), such that they realise $\alpha_{\min}(G)$, and when we assign color one to vertices from the largest independent set including $v_1$, the second vertex, $v_2$ forms maximal independent set, of size 1. So, we have to partition the rest of vertices into minimum number of independent sets of size at most 2. In consequence, by adding vertices of degree 2 that are adjacent to $v_1$ and $v_2$, vertices $u_1, \ldots,u_n$, we are able to create arbitrary large chordal graphs with large equitable chromatic number.
\begin{figure}
\centering
\includegraphics[scale=1]{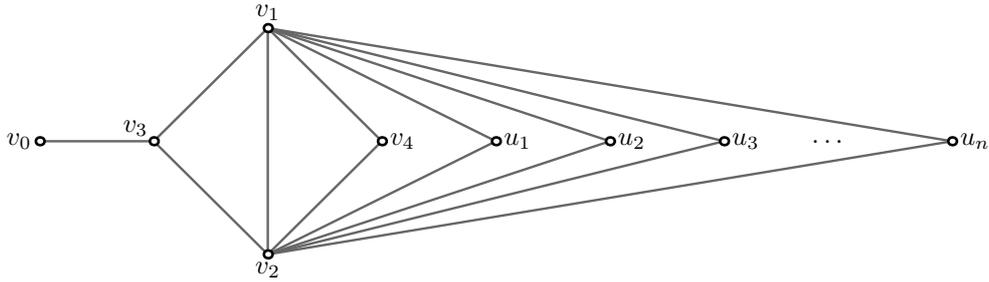}
\caption{An example of the construction of chordal graphs with unbounded value of $\chi_=(G)$.}
\label{ex_gen}
\end{figure}
Note, that $\chi_=(G)=3+\lceil n/2 \rceil$, while $\max\left\{\omega(G), \left\lceil \frac{|V(G)|+1}{\alpha_{min}(G)+1}\right \rceil\right\}+1=\lceil n/3\rceil+3$. 
Thus, the difference between $\chi_=(G)$ and the maximum from Conjecture \ref{conj:gap1} can be arbitrary large for general chordal graphs.

From our perspective, proving Conjecture \ref{conj:gap1} and the equality $\chi_=(G)=\chi_=^*(G)$ for block graphs seem promising for future work. It seems also desirable to prove inequalities from Conjecture \ref{conj:gap1} for other interesting graph classes. In other words, we would like to find other graph classes where the bounds offered by Conjecture \ref{conj:gap1} are going to hold.



\bibliographystyle{elsarticle-num}


\end{document}